\documentclass[a4paper, 12pt]{article}

\usepackage[authoryear, sort&compress]{natbib}
\bibpunct{(}{)}{;}{a}{}{,} 

\usepackage{amsthm, amsmath, amssymb, mathrsfs, multirow, url}
\usepackage{subfigure}
\usepackage{ifthen} 
\usepackage{amsfonts}
\usepackage[usenames]{color}
\usepackage{fullpage}

\usepackage{graphicx} 
\usepackage{bm}
\usepackage{apalike} 
\usepackage[ruled]{algorithm2e}
\usepackage{enumerate}
\usepackage{enumitem}
\usepackage{multirow}

\newcommand{\Inv}{^{-1}} 
\DeclareMathSymbol{\Gamma}{\mathord}{operators}{"00}
\DeclareMathSymbol{\Delta}{\mathord}{operators}{"01}
\DeclareMathSymbol{\Theta}{\mathord}{operators}{"02}
\DeclareMathSymbol{\Lambda}{\mathord}{operators}{"03}
\DeclareMathSymbol{\Xi}{\mathord}{operators}{"04}
\DeclareMathSymbol{\Pi}{\mathord}{operators}{"05}
\DeclareMathSymbol{\Sigma}{\mathord}{operators}{"06}
\DeclareMathSymbol{\Upsilon}{\mathord}{operators}{"07}
\DeclareMathSymbol{\Phi}{\mathord}{operators}{"08}
\DeclareMathSymbol{\Psi}{\mathord}{operators}{"09}
\DeclareMathSymbol{\Omega}{\mathord}{operators}{"0A}
\newcommand\smallop{
  \mathchoice
    {{\scriptstyle\mathcal{O}_p}}
    {{\scriptstyle\mathcal{O}_p}}
    {{\scriptscriptstyle\mathcal{O}_p}}
    {\scalebox{.7}{$\scriptscriptstyle\mathcal{O}_p$}}
  }

\theoremstyle{plain} 
\newtheorem{thm}{Theorem}

\theoremstyle{definition}

\theoremstyle{remark}

\newcommand{\prob}{\mathsf{P}} 

\newcommand{\cov}{\mathsf{C}}

\newcommand{\nm}{{\sf N}}
\newcommand{\expo}{{\sf Exp}}

\newcommand{\sn}{{\sf SN}}

\newcommand{\RR}{\mathbb{R}}

\newcommand{\TT}{\mathbb{T}}

\newcommand{\eps}{\varepsilon}

\newcommand{\sens}{\mathsf{sens}}
\newcommand{\spec}{\mathsf{spec}}

\title{Model-free posterior inference on the area under the receiver operating characteristic curve}
\author{Zhe Wang\footnote{Department of Statistics, North Carolina State University; {\tt zwang54@ncsu.edu}, {\tt rgmarti3@ncsu.edu}} \quad and \quad Ryan Martin$^*$}
\date{\today}

\begin{document}

\maketitle 

\begin{abstract}   
The area under the receiver operating characteristic curve (AUC) serves as a summary of a binary classifier's performance. 
Methods for estimating the AUC have been developed under a  binormality assumption which restricts the distribution of the score produced by the classifier.  However, this assumption introduces an infinite-dimensional nuisance parameter and can be inappropriate, especially in the context of machine learning.  This motivates us to adopt a model-free Gibbs posterior distribution for the AUC.  We present the asymptotic Gibbs posterior concentration rate, and a strategy for tuning the learning rate so that the corresponding credible intervals achieve the nominal frequentist coverage probability. Simulation experiments and a real data analysis demonstrate the Gibbs posterior's strong performance compared to existing methods based on a rank likelihood.  

\smallskip

\end{abstract}

\section{Introduction}
\label{S:intro}  
First proposed during World War II to assess the performance of radar receiver operators \citep{cali2015some}, the receiver operating characteristic (ROC) curve is now an essential tool for analyzing the performance of binary classifiers in areas such as  signal detection \citep{green1966signal}, psychology examination \citep{swets1973relative,swets1986indices}, radiology \citep{lusted1960logical,hanley1982meaning}, medical diagnosis \citep{swets1982evaluation,hanley1989receiver}, and data mining \citep{spackman1989signal,fawcett2006introduction}. One informative summary of the ROC curve is the corresponding area under the curve (AUC).  This measure provides an overall assessment of classifier's performance, independent of the choice of threshold, and is, therefore, the preferred method for evaluating classification algorithms \citep{provost1997analysis,provost1998case,bradley1997use,huang2005using}. The AUC is an unknown quantity, and our goal is to use the information contained in the data to make inference about the AUC.  
The specific set up is as follows. For a binary classifier which produces a random \emph{score} to indicate the propensity for, say, Group~1; individuals with scores higher than a threshold are classified to Group 1, the rest are classified to Group~0. Let $U$ and $V$ be independent scores corresponding to Group~1 and Group~0, respectively. Given a threshold $t$, define the \emph{specificity} and \emph{sensitivity} as $\spec(t)=\prob(V<t)$ and $\sens(t)=\prob(U>t)$. Then the ROC curve is a plot of the parametric curve $\bigl(1-\spec(t),\sens(t)\bigr)$ as $t$ takes all possible values for scores. While the ROC curve summarizes the classifier's tradeoff between sensitivity and specificity as the threshold varies, the AUC measures the probability of correctly assigning scores for two individuals from two groups, which equals $\prob(U>V)$ \citep{bamber1975area}, and is independent of the choice of threshold. Consequently, the AUC is a functional of the joint distribution of $(U,V)$, denoted by $\prob$, so the ROC curve is actually not needed to identify AUC.  



In the context of inference on the AUC, when the scores are continuous, it is common to assume that $\prob$ satisfies a so-called {\em binormality} assumption, which states that there exists a monotone increasing transformation that maps both $U$ and $V$ to normal random variables  \citep{hanley1988robustness}.  For most medical diagnostic tests, where the classifiers are simple and ready-to-use without training, such an assumption serves well \citep{hanley1988robustness,metz1998maximum,cai2004semi}, although it has been argued that other distributions can be more appropriate for some specific tests \citep[e.g.,][]{guignard1983validity,goddard1990receiver}.  But for complicated classifiers which involve multiple predictors, as often arise in machine learning applications,  binormality---or any other model assumption for that matter---becomes a burden.  This motivates our pursuit of a ``model-free'' approach to inference about the AUC.

Specifically, our goal is the construction of a type of posterior distribution for the AUC.  The most familiar such construction is via Bayes's formula, but this requires a likelihood function and, hence, a statistical model.  The only way one can be effectively ``model-free'' within a Bayesian framework is to make the model extra flexible, which requires lots of parameters.  In the extreme case, a so-called Bayesian nonparametric approach would take the distribution $\prob$ itself as the model parameter.  When the model includes lots of parameters, then the analyst has the burden of specifying prior distributions for these, based on little or no genuine prior information, and also computation of a high-dimensional posterior.  But since the AUC is just a one-dimensional feature of this complicated set of parameters, there is no obvious return on the investment into prior specification and posterior computation.  A better approach would be to construct the posterior distribution for the AUC directly, using available prior information about the AUC only, without specifying a model and without the introduction of artificial model parameters.  That way, the data analyst can avoid the burdens of prior specification and posterior computation, bias due to model misspecification, and issues that can arise as a result of non-linear marginalization \citep[e.g.,][]{martin2016false,fraser2011bayes}.  

As an alternative to the traditional Bayesian approach, we consider here the construction of a so-called {\em Gibbs posterior} for the AUC.  In general, the Gibbs posterior construction proceeds by defining the quantity of interest as the minimizer of a suitable risk function, treating an empirical version of that loss function like a negative log-likelihood, and then combining with a prior distribution very much like in Bayes's formula.  General discussion of Gibbs posteriors can be found in \citet{zhang2006eps, zhang2006information}, \citet{bissiri2016general} and \citet{alquier2016properties}; statistical applications are discussed in \citet{jiang2008gibbs} and \citet{syring2017gibbs, syring2017calibrating, syring.martin.image}.   Again, the advantage is that the Gibbs posterior avoids model misspecification bias and the need to deal with unimportant nuisance parameters.  Moreover, under suitable conditions, Gibbs posteriors can be shown to have desirable asymptotic concentration properties, with theory that parallels that of Bayesian posteriors under model misspecification \citep[e.g.,][]{kleijn2006misspecification, kleijn2012bernstein}.  

A subtle point is that, while the risk minimization problem that defines the quantity of interest is independent of the scale of the loss function, the Gibbs posterior is not.  This scale factor is often referred to as the {\em learning rate}  \citep[e.g.,][]{grunwald2012safe} and, because it controls the spread of the Gibbs posterior, its specification needs to be handled carefully.  Various approaches to the specification of the learning rate parameter  \citep[e.g.,][]{grunwald2012safe,grunwald2017inconsistency,bissiri2016general,holmes2017assigning,lyddon2019general}.  Here we adopt the approach in
\citet{syring2017calibrating} that aims to set the learning rate so that, in addition to its robustness to model misspecification and asymptotic concentration properties, the Gibbs posterior credible sets have the nominal frequentist coverage probability. 
When the sample size is large, we recommend an (asymptotically) equivalent calibration method that is simpler to compute.

The present paper is organized as follows. In Section~\ref{S:binormality}, we review some methods for making inference on the AUC based on the binormality assumption, in particular, the Bayesian approach in \citet{gu2009bayesian} that involves a suitable rank-based likelihood. 
In Section~\ref{S:counter example}, we argue that the binormality assumption is generally inappropriate in machine learning applications, and provide one illustrative example involving a support vector machine.  This difficulty with model specification leads us to the Gibbs posterior, a model-free alternative to a Bayesian posterior, which is reviewed in Section~\ref{SS:gibbs}. 
We develop the Gibbs posterior for inference on the AUC, derive its asymptotic concentration properties, and investigate how to properly scale the risk function in Section~\ref{S:main}. Simulation experiments are carried out in Section~\ref{S:examples}, where a Gibbs posterior estimator performs favorably compared with the Bayesian approach based on a rank-based likelihood.  We also apply the Gibbs posterior on a real dataset for evaluating the performance of a biomarker for pancreatic cancer and compare our result with those based on the rank likelihood. Finally, we give some concluding remarks in Section~\ref{S:discuss}.

\section{Background}
\label{S:background}

\subsection{Binormality and related methods}
\label{S:binormality}
Following \citet{hanley1988robustness}, the scores $U$ and $V$ satisfy the binormality assumption if their distribution functions
are
$\Phi[b^{-1}\{H(u) - a\}]$ and $\Phi\{H(v)\}$ respectively, where $a>0$, $b>0$, $H$ is a monotone increasing function, and $\Phi$ denotes the $\nm(0,1)$ distribution function, which implies 
that $U$ and $V$ can be transformed to $\nm(a,b^2)$ and $\nm(0,1)$ via $H$.
If $\prob=\prob_{a,b,H}$ denotes the distribution of $(U,V)$ under this assumption, then  
the ROC curve and the AUC, respectively, are given by
$t \mapsto \Phi[b^{-1}\{a+\Phi\Inv(t)\}]$ and 
\begin{align}\label{eq:binorm_AUC}
    \Phi\{a(b^2+1)^{-1/2}\}.
\end{align} 
Even though $H$ is not needed to define the AUC---only $(a,b)$---since the joint distribution of $(U,V)$ does depend on $H$, any likelihood-based method would have to deal with this infinite-dimensional nuisance parameter.  Some strategies are used to avoid dealing with $H$ directly.  The semi-parametric approach in \citet{cai2004semi} manipulates the equivalent
densities ratio of $U$ over $V$ and  $W$ over $Z$, and introduces cumulative hazard function as a nuisance parameter. A profile likelihood is obtained based on a discrete estimate for the cumulative hazard function.
In the approach of \citet{metz1998maximum}, data are suitably grouped and a multinomial pseudo-likelihood is constructed.  Alternatively, since data ranks are invariant to monotone transformations, one can construct a rank-based likelihood, as in \citet{zou2000two}, which can be maximized over $(a,b)$ to estimate the AUC.  But it turns out that a Bayesian approach that uses Monte Carlo sampling from a rank-based posterior distribution, as in \citet{gu2009bayesian}, is computationally more efficient than maximizing the rank likelihood.  Since this is our proposed method's primary competitor, we give some details about Gu and Ghosal's Bayesian rank-based likelihood approach here.  

Consider the transformed scores $W=H(U)$ and $Z=H(V)$, according to the binormality assumption, its joint distribution can be written as $\prob_{a,b}$, no more dependence on $H$.
Elimination of the nuisance parameter $H$ is desirable, but $(W,Z)$ are unavailable to us without knowledge of $H$.  That is, unless we consider a function of $(U,V)$ that is invariant to transformations by $H$.  A good candidate function is the {\em ranks}.  That is, let $R_{U,V}$ denote the ranks of the vector $(U_1,\ldots,U_m, V_1, \ldots, V_n)$, where $(U_1,\ldots,U_m)$ and $(V_1,\ldots,V_n)$ are independent and identically distributed (iid) copies of $U$ and $V$, respectively.  Then 
\begin{equation}
\label{eq:rank likeli}
\prob_{a,b,H}(R_{U,V}=r) \equiv \prob_{a,b}(R_{W,Z}=r), 
\end{equation}
where $R_{W,Z}$ is the ranks of $(W_1,\ldots,W_m,Z_1,\ldots,Z_n)$, with $W_i=H(U_i)$ and $Z_j = H(V_j)$.  The key is that the observed ranks based on the $(U_i, V_j)$ sample can be plugged in for $r$ on the right-hand side of \eqref{eq:rank likeli} and that gives a likelihood function for $(a,b)$, without requiring knowledge of $H$.  Of course, this is not a proper likelihood function, i.e., there is loss of information caused by throwing away the values of $(U_i, V_j)$, but eliminating the infinite-dimensional nuisance parameter might be worth the price, especially when the goal is inference on the ROC curve or AUC, neither of which depend directly on $H$.  
The approach outlined in \citet{gu2009bayesian} proceeds by treating the $(W_i,Z_j)$ values as latent variables and defining a full posterior for $(a, b^2, W_i,Z_j)$, given $R_{U,V}$, and then marginalizing out $(W_i,Z_j)$ to get a posterior distribution for $(a,b^2)$ alone.  
If we take the Jeffreys prior for $(a,b^2)$, which is proportional to $b^{-2}$, then the full conditional distribution presented in \citet{gu2009bayesian} are 
\begin{align*}
(a \mid W_1,\dots, W_m,Z_1,\dots, Z_n, b^2, R_{U,V}) & \sim 
\nm\bigl( m^{-1} \textstyle\sum_{i=1}^m W_i, b^2 m^{-1} \bigr), \\
(b^2 \mid W_1,\dots, W_m,Z_1,\dots, Z_n, a, R_{U,V}) & \sim 
{\sf IG} \bigl( \tfrac{m-1}{2} ,\tfrac{1}{2} \textstyle\sum_{i=1}^m (W_i-a)^2 \bigr), \\
(W_i\mid W_{-i}, Z_1,\ldots,Z_n, a, b^2, R_{U,V}) & \sim \nm(a, b^2) \times 1(R_{W,Z} = R_{U,V}), \quad i=1\ldots m \\
(Z_j\mid Z_{-j}, W_1,\ldots,W_m,a,b^2,R_{U,V}) & \sim \nm(0, 1) \times 1(R_{W,Z} = R_{U,V}), \quad j=1\ldots n
\end{align*}
where, e.g., $W_{-i}=(W_1,\ldots,W_{i-1},W_{i+1},\ldots,W_m)$, ${\sf IG}(\alpha,\beta)$ denotes the inverse gamma distribution with density $\frac{\beta^\alpha}{\Gamma(\alpha)} x^{-\alpha-1} e^{-\beta /x}$, and $1(\cdot)$ denotes the indicator function. With these full conditionals, it is straightforward to develop a Monte Carlo strategy that produces samples from the $(a,b^2)$ posterior distribution.  These samples can then be used to get a posterior distribution for AUC using the expression in \eqref{eq:binorm_AUC}.

\subsection{Validity of binormality in machine learning applications}
\label{S:counter example}

Before the ROC and AUC analysis were introduced to machine learning area, the binormality assumption had been proposed and used in the context of medical diagnosis for {\em simple classifiers}, where the scores $U$ and $V$ are determined based on a single predictor variable.  When assuming binormality for classifiers in machine learning \citep[e.g.,][]{brodersen2010binormal,macskassy2004confidence}, the situation differs because multiple predictor variables are usually involved.

Suppose that a binary $Y\in\{0,1\}$ indicates the group, $X\in \mathbb{R}^p$ is the predictors, and $\mathcal{M}_{X,Y}$ denotes the joint distribution of $(X, Y)$.
As described in Section \ref{S:intro}, a binary classifier provides a parametric form of the predictor, namely the \emph{score} $S(X;\beta)$, to indicate the propensity for $Y$ taking value 1. 
A training process is generally needed for estimating
the unknown $\beta$ based on a set of observations $\{X_i, Y_i\}_1^n$.
Let the estimator be denoted as $\hat{\beta}_n=\hat{\beta}_n(\{X_i,Y_i\}_1^n) \in \RR^p$. 
It follows that the random score $U$ for Group~1 is defined as  $S(X;\hat{\beta}_n)$ where the predictor $X$ follows the conditioned distribution $\mathcal{M}_{X|Y=1}$. Similarly, the random score $V$ for Group~0 is defined as  $S(X;\hat{\beta}_n)$ where $X\sim\mathcal{M}_{X|Y=0}$.
By assuming that $\hat{\beta}$ converges to a non-random quantity when the size of the training set goes to infinity, $U$ and $V$ are asymptotically independent.

The binormality assumption for simple classifiers, which are special cases where $p=1$ and $S(X;\beta)=X$,
only requires the $\mathcal{M}_{X|Y=1}$ and $\mathcal{M}_{X|Y=0}$ to be normals (after the transformation $H$).
In the general case, where $p>1$, even if every one of the $p$ predictors obeys the binormality assumption, the scores for two groups are still not guaranteed to satisfy the binormality, since $S(X;\beta)$ can take virtually any form.  

For example, consider two 
independent and identically distributed predictors $X_1$ and $X_2$,
given different groups ($Y=1$ or $0$), the predictors are distributed as $\nm(0,2)$ or $\nm(0,50)$, respectively.
For training data (Figure~\ref{fig:svm_example}(a)), $10,000$ copies of $(X_1,X_2,Y)$ are generated for each group. 
A support vector machine with radial basis function kernel is applied to this non-linearly separable dataset
and correspondingly the predicted scores for another $10,000$ new data copies under the same data generating scheme are recorded as $\{U_1,\ldots, U_m\}$ and $\{V_1,\ldots, V_n\}$.
Then the unique monotone increasing transformation $H$ which transforms $V$ to be standard normal 
is approximated by $H = \Phi^{-1} \circ \hat F_V$, where 
$\hat{F}_V(v)=n^{-1}\sum_{j=1}^n 1(V_j\leq v)$ is the empirical distribution.
The histogram of $H(U)$ in Figure~\ref{fig:svm_example}(b) does not agree with the fitted normal density. And a Q-Q plot in Figure~\ref{fig:svm_example}(c) for $U$ samples also suggest 
there is no such $H$ which transforms $U$ and $V$ to a model that satisfies the binormality assumption.


\begin{figure}
  \centering
  \subfigure[]{\includegraphics[width=0.3\textwidth,height=5.5cm]{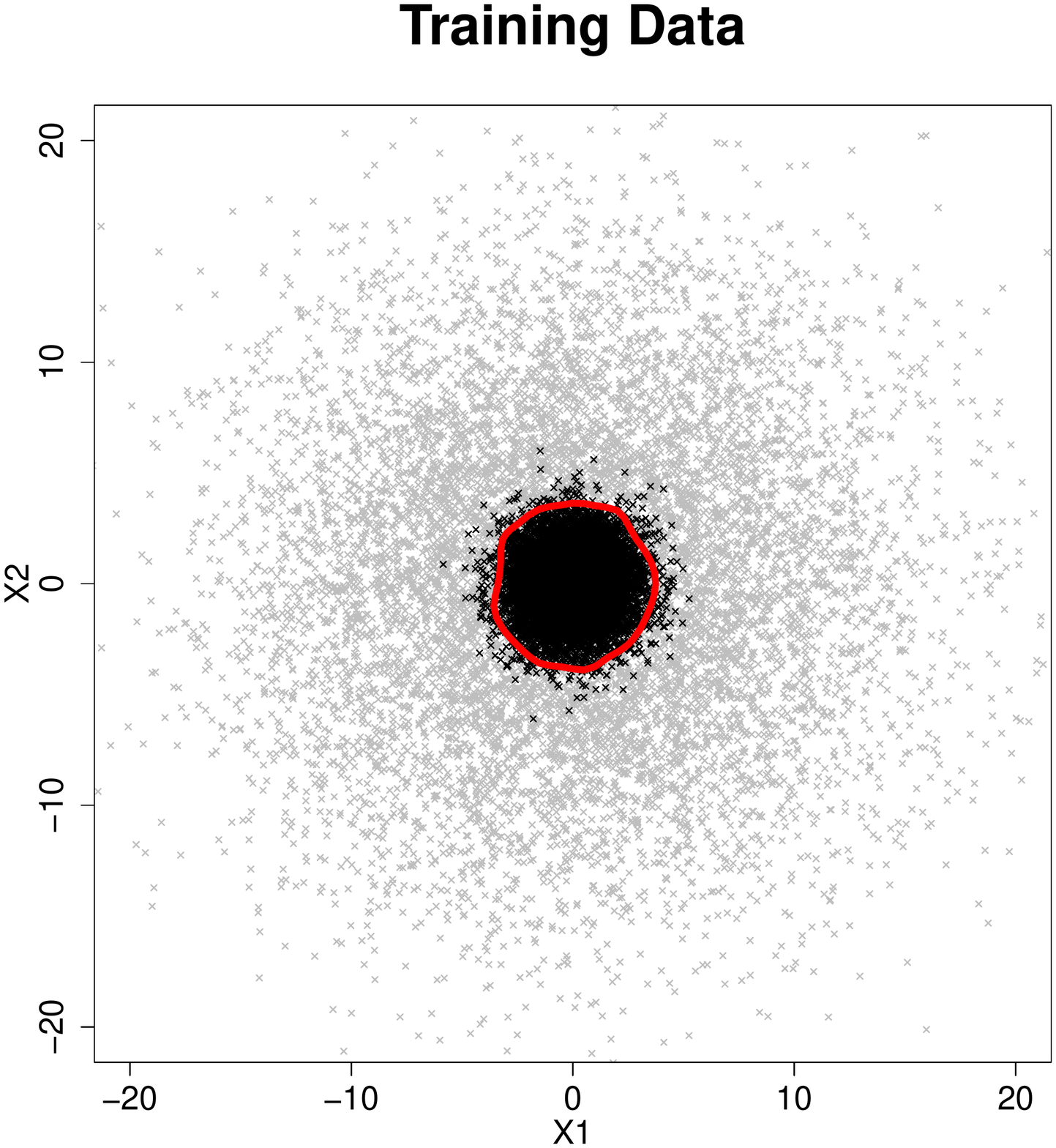}} 
  \subfigure[]{\includegraphics[width=0.3\textwidth,height=5.5cm]{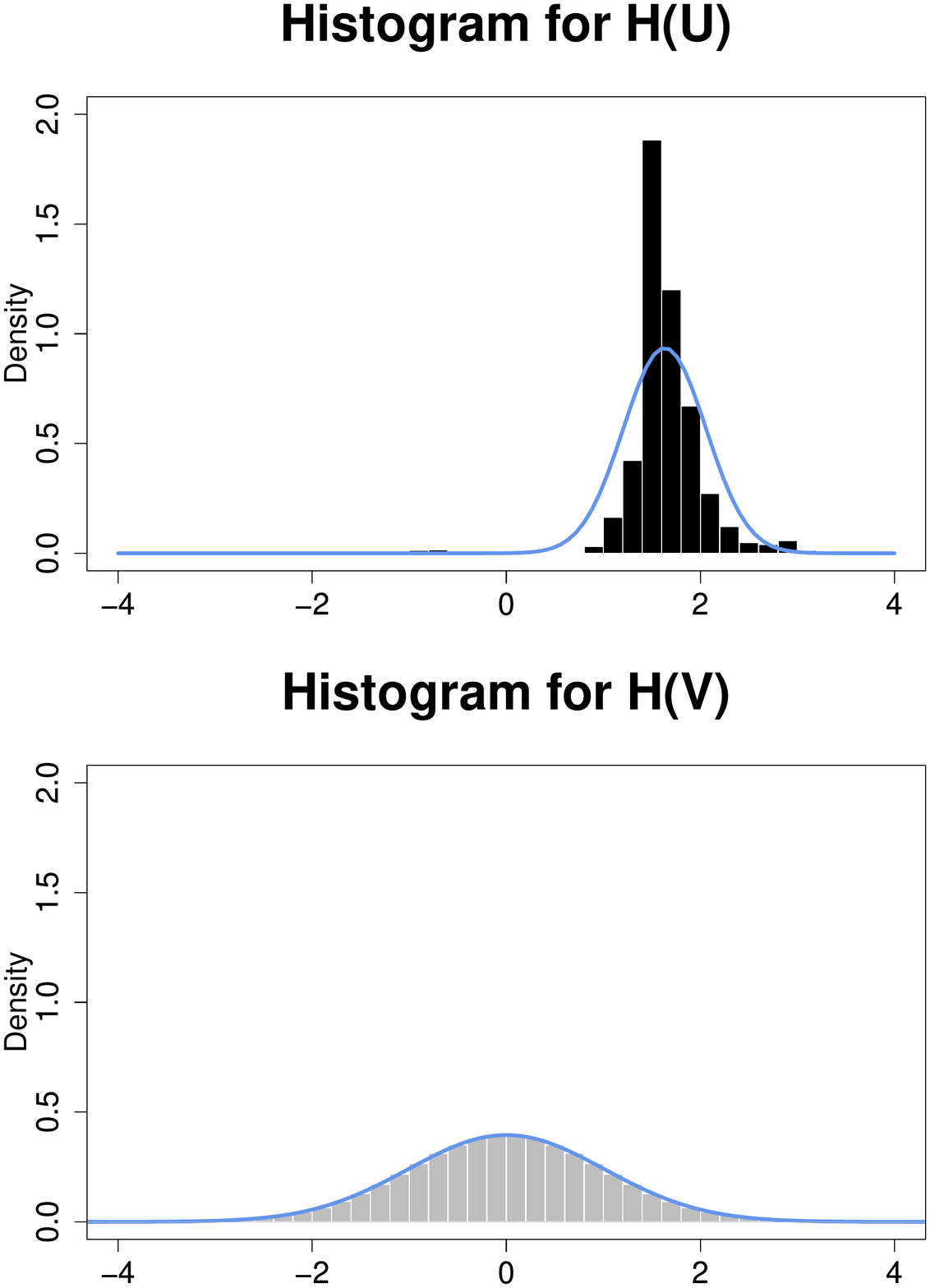}}
  \subfigure[]{\includegraphics[width=0.3\textwidth,height=5.5cm]{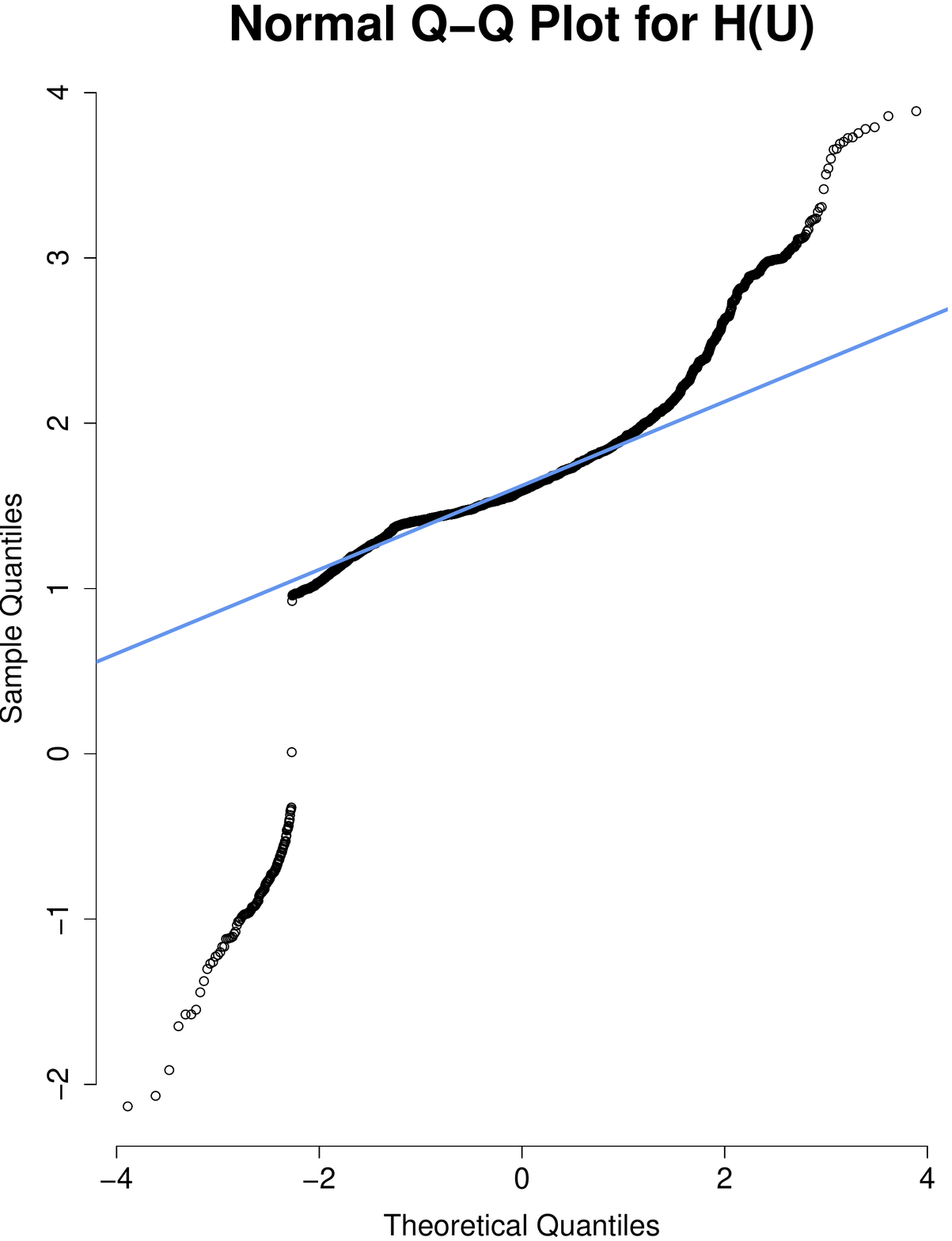}}
  \caption{
  (a) Training data and the SVM decision boundary (red curve). For each point, two predictors  $(X_1,X_2)$ are plotted on axes and $Y$ is visualized by the color (black for Group~1, gray for Group~0); 
  (b) Histograms for $H(U)$ and $H(V)$ with the fitted normal densities;
  (c) Q-Q plot for transformed samples $\{H(U_1),\ldots, H(U_m)\}$}
  \label{fig:svm_example}
\end{figure}
   
\subsection{Gibbs posterior distributions}
\label{SS:gibbs}

A Gibbs posterior distribution resembles a Bayesian posterior, but is constructed using different ingredients.  In particular, the Gibbs posterior does not start with a statistical model and likelihood, it starts with a more general connection between data and quantities of interest, through a loss function.  Suppose that data $T_1,\ldots,T_n$ are identically distributed $\TT$-valued observations from distribution $\prob$, and that there is some functional $\theta=\theta(\prob)$, taking values in $\Theta$, about which inference is desired.  Instead of introducing a statistical model for $\prob$---that is, assuming $\prob$ takes a particular distributional form $\prob_\zeta$ for some model parameter $\zeta$, and then expressing $\theta$ as a function of $\zeta$---we construct a posterior for $\theta$ directly as follows.  Assume that there exists a loss function $\ell_\theta(t)$, mapping $\TT \times \Theta$ to $\RR$, such that the true value, $\theta^\star$, of $\theta$ solves the optimization problem 
\begin{equation}
\label{eq:truth}
\theta^\star = \arg\min_\theta R(\theta), 
\end{equation}
where the risk function $R(\theta) = \prob \ell_\theta$ is just the expected loss with respect to $\prob$.  When the quantity of interest is defined as the solution to an optimization problem, it makes sense to estimate that quantity by solving an empirical version of the optimization problem, 
\[ \hat\theta_n = \arg\min_\theta R_n(\theta), \]
where the empirical risk $R_n(\theta) = \widehat\prob_n \ell_\theta$ is the expected loss with respect to the empirical distribution $\widehat\prob_n = n^{-1} \sum_{i=1}^n \delta_{T_i}$, with $\delta_t$ the point-mass distribution concentrated at $t$. 
From this empirical risk function, the Gibbs posterior distribution is defined as 
\begin{align}
\label{eq:GP_general}
\Pi_n(d\theta) \propto e^{-\omega n R_n(\theta)} \, \Pi(d\theta), \quad \theta \in \Theta, 
\end{align}
where $\Pi$ is a prior distribution on $\Theta$ and $\omega > 0$ is a scale parameter to be determined; see \citet{bissiri2016general} for the decision-theoretic underpinnings of this approach.  

For us, the motivation behind the use of a Gibbs posterior is that it gives us direct, model-free posterior inference about the quantity of interest.  This is beneficial because, for one thing, a statistical model could be misspecified and that would generally bias the results.  But even if the model is correctly specified, it is unlikely that an appropriate statistical model could be described in terms of $\theta$ alone, so the model index $\zeta$ would include a number of nuisance parameters that require prior distribution specification and posterior computation, efforts that are effectively wasted if marginal inference on $\theta$ is the goal.  The Gibbs posterior, by targeting $\theta$ directly, avoids the possible misspecification bias, allows for prior beliefs about $\theta$ to be readily accommodated, and does not require dealing with nuisance parameters.  And the applications presented in \citet{syring2017gibbs, syring2017calibrating, syring.martin.image}, along with the one presented here, suggest that this direct approach has a number of important advantages over the more traditional Bayesian counterpart.  

Of course, the magnitude of the loss function does not affect the solution to the optimization problem in \eqref{eq:truth}, nor that in the empirical version thereof.  But the magnitude does affect the Gibbs posterior in \eqref{eq:GP_general}, which is why we include the scaling factor $\omega$.  Data-driven strategies for specifying this tuning parameter are discussed in Section~\ref{SS:calibrate} below.


\section{Gibbs posterior for the AUC}
\label{S:main}

\subsection{Definition}
As mentioned, the AUC is a functional of the joint distribution $\prob$ of $(U,V)$, i.e., $\theta = \theta(\prob)$, given by $\theta = \prob(U > V)$.  Recall that the data consists of independent copies $(U_1,\ldots,U_n)$ and $(V_1,\ldots,V_m)$ of $U$ and $V$, respectively.  To construct a Gibbs posterior distribution for $\theta$ as discussed above, we need an appropriate loss function.  That is, we need a function $\ell_\theta(u,v)$ such that the corresponding risk function, $R(\theta) = \prob \ell_\theta$, is minimized at the true AUC, $\theta^\star$.  If we define 
\[ \ell_\theta(u,v) = \{\theta - 1(u > v)\}^2, 
\quad \theta \in [0,1], \]
then it is easy to check that 
\[ R(\theta) = \theta^2 -2 \theta^\star \theta + \theta^{\star 2}, \]
and, moreover, that this risk function is uniquely minimized at $\theta=\theta^\star$.  Then the empirical risk function is
\[ R_{m,n}(\theta) = \widehat{\prob}_{m,n} \ell_\theta = \frac{1}{mn} \sum_{i=1}^m \sum_{j=1}^n \{\theta - 1(U_i > V_j)\}^2 \]
where $\widehat{\prob}_{m,n}=(mn)^{-1}\sum_{i=1}^m \sum_{j=1}^n\delta_{(U_i,V_j)}$ is the empirical distribution of the score pairs. Note that the minimizer of the empirical risk function, namely, 
\begin{equation}
\label{eq:Mann-Whitney}
\hat\theta_{m,n} = \arg\min_\theta R_{m,n}(\theta) = \frac{1}{mn} \sum_{i=1}^n \sum_{j=1}^m 1(U_i > V_j),
\end{equation}
is the familiar statistic suggested by \citet{mann1947test} for testing if one of two independent random variables is stochastically larger than the other.  

Following the general approach described in Section~\ref{SS:gibbs}, we can construct a Gibbs posterior distribution for the AUC, with density 
\[ \pi_{m,n}(\theta) \propto e^{-\omega mn R_{m,n}(\theta)} \, \pi(\theta), \quad \theta \in [0,1], \]
where $\pi$ is some prior density for the AUC, and $\omega$ is the learning rate to be specified in Section~\ref{SS:calibrate}.  This Gibbs posterior does not require any model assumptions, does not require marginalization over nuisance parameters, and can directly incorporate available prior information about $\theta$.  Moreover, the Gibbs posterior is approximately centered around $\hat\theta_{m,n}$, which is a quality estimator of the AUC, regardless of what form the underlying distribution $\prob$ takes, so we can expect the Gibbs posterior---for suitable $\omega$---to provide quality model-free inference.  Details on the asymptotic concentration properties of the Gibbs posterior are presented in the next section.  

After some simple algebra, the Gibbs posterior above can be re-expressed as 
\begin{equation}
\label{eq:GP_AUC}
\pi_{m,n}(\theta) \propto e^{-\omega mn (\theta - \hat\theta_{m,n})^2} \pi(\theta), \quad \theta \in [0,1], 
\end{equation}
which shows some resemblance to a truncated normal distribution.  A very reasonable choice of prior is a truncated normal distribution with informative choices of prior location $\mu_0$ and scale $\sigma_0$.  With this choice, the Gibbs posterior is a truncated normal distribution too, with corresponding location and scale, respectively,  
\[ \mu_{m,n} = \frac{\mu_0 + 2\omega \sigma_0^2 mn \hat\theta_{m,n}}{1 + 2 \omega \sigma_0^2 mn} \quad \text{and} \quad \sigma_{m,n} = \Bigl\{ \frac{\sigma_0^2}{1 + 2 \omega \sigma_0^2 mn} \Bigr\}^{1/2}. \]
In the absence of prior information about the AUC, one can take a flat uniform prior, $\pi(\theta) \equiv 1$, in which case the Gibbs posterior is still a truncated normal distribution but with location and scale, respectively, 
\[ \mu_{m,n} = \hat\theta_{m,n} \quad \text{and} \quad \sigma_{m,n} = (2\omega m n)^{-1/2}. \]
In practice, we recommend the use of any available prior information about the AUC whenever possible, but, for the rest of this paper, we will work with the Gibbs posterior based on the default uniform prior.  

Here we are concerned with inference on AUC for a given classifier, and consequently the posterior is constructed directly for the AUC.  \citet{ridgway2014pac} also construct a Gibbs posterior using AUC, but their goal is to find a classifier that maximizes AUC.  

\subsection{Asymptotic concentration properties}
\label{S:converge rate}

It is natural to ask what kind of asymptotic concentration properties the Gibbs posterior distribution enjoys.  An advantage of our approach's simplicity is the ease in which the convergence properties can be deduced, but some care is needed in formulating the asymptotic regime precisely.  Indeed, since the two groups may have different sample sizes, it is clear that what we need is for the {\em smaller} of the two sample sizes to go to infinity.  Therefore, the rate is determined by $m \wedge n$, and following theorem states that, under no conditions on the joint distribution $\prob$ of $(U,V)$, the Gibbs posterior distribution concentrates asymptotically around the true AUC at the rate $(m \wedge n)^{-1/2}$.  



\begin{thm}
\label{thm:rate}
Let $\theta^\star$ be the true AUC corresponding to the joint distribution $\prob$, and assume, without loss of generality, that $n = m \wedge n$.  If $\Pi_{m,n}$ is the Gibbs posterior defined in \eqref{eq:GP_AUC} based on a fixed learning rate $\omega > 0$ and a prior density $\pi$ that is positive and continuous in an interval containing $\theta^\star$, then for any sequence $K_n \to \infty$, 
\[ \Pi_{m,n}(\{\theta:|\theta-\theta^\star|> K_{n} \,(m\wedge n)^{-1/2} \}) \to 0 \quad \text{in $\prob$-probability as $n \to \infty$}. \]
\end{thm}

\begin{proof}
See Appendix~A. 
\end{proof}

Several remarks on the concentration rate theorem, its consequences, and some related results are in order.  
\begin{itemize}
\item 
The convergence in $\prob$-probability conclusion in Theorem~\ref{thm:rate} can be strengthened to convergence with $\prob$-probability $1$ by assuming that sample sizes for two groups increase at the same rate, i.e., $m(m+n)^{-1} \to \rho\in(0,1)$.  Under this condition, \citet[][Chap.~3.2]{korolyuk2013theory} show that $\hat\theta_{m,n} \to \theta^\star$ with $\prob$-probability 1 and, with this, the stronger Gibbs posterior concentration rate result can be proved along lines similar to those in Appendix~A below.  

\item 
As shown in \eqref{eq:GP_general}, the Gibbs posterior resembles a Bayesian posterior based on a suitably misspecified model, one whose ``likelihood function'' equals $\exp\{-\omega n R_n(\theta)\}$.  Even in misspecified cases, Bernstein--von Mises-style distributional approximations are possible; see, e.g., \citet{kleijn2012bernstein}.  In our case, we immediately see a truncated normal form of the Gibbs posterior, so as long as $\theta^\star$ is in the interior of $(0,1)$, the asymptotic normality of the Gibbs posterior is automatic.  

\item
We note the loss scale $\omega$ controls 
the proportion of information in the Gibbs posterior which is learned from the data. Consequently, it is reasonable to adjust $\omega$ so that a set of observations with a larger size is given more trust. In fact, if we substitute the fixed $\omega$ in Theorem~\ref{thm:rate} with a sequence $\omega_n$ that vanishes slower than $(m \vee n)^{-1}$, 
then the Gibbs posterior concentration rate result still holds.  
\end{itemize}

\subsection{Tuning the learning rate}
\label{SS:calibrate}

The good behavior of a Bayesian posterior is guaranteed only when the model is correctly specified. Under misspecification, even if the posterior concentrate around an efficient estimator, the asymptotic variance of the posterior could be drastically different from that of the efficient estimator; see \citet{kleijn2012bernstein}. 
Consequently, $100(1-\alpha)$\% credible regions from a misspecified Bayes model may not achieve the nominal $100(1-\alpha)$\% confidence, even asymptotically.  
Fortunately, the Gibbs posterior learning rate parameter, $\omega$, which controls the spread, can be tuned in such a way that this undesirable discrepancy between credibility and confidence is avoided.  Various tuning strategies are available in the literature \citep[e.g.,][]{bissiri2016general, fasiolo2017fast, lyddon2019general, grunwald2012safe}, but only the approach presented in \citet{syring2017calibrating} focuses directly on coverage probability, so that is the approach we will adopt here.  

Algorithm~1 describes the calibrating procedure from \citet{syring2017calibrating} in the context of inference on the AUC.  The rationale behind this algorithm is as follows.  Take a $100(1-\alpha)$\% credible interval based on the Gibbs posterior \eqref{eq:GP_AUC} with learning rate $\omega$, in particular, the highest posterior density credible interval.  Then the frequentist coverage probability of that credible interval, call it $c_\alpha(\omega)$, depends on $\omega$, $\alpha$, and other things.  If we could evaluate $c_\alpha(\omega)$, that is, if we knew and could directly simulate from $\prob$, then we could just solve the equation $c_\alpha(\omega) = 1-\alpha$.  For future reference, in this ideal case, we call the solution to this equation the {\em oracle learning rate}.  In real applications, however, $\prob$ is unknown, so we cannot evaluate $c_\alpha(\omega)$ exactly, but we can get an estimate using the bootstrap, and then solve that equation using stochastic approximation \citep{robbins1951stochastic} with step size sequence $(\kappa_t)$ that satisfies
\begin{equation}
\label{eq:steps}
\textstyle\sum_{t=1}^\infty \kappa_t = \infty \quad \text{and} \quad \textstyle\sum_{t=1}^\infty \kappa_t^2 < \infty
\end{equation}
Details are discussed in \citet{syring2017calibrating}.  


The method implemented in Algorithm~1 requires the repeated processing of bootstrap samples and, therefore, can be computationally expensive when the sample sizes are large.  For such cases, however, there is an alternative strategy, based on ideas in \citet{lyddon2019general}, that is both easier and faster, while still providing approximate calibration in the sense above.  The idea is that we want the Gibbs posterior variance to be roughly equal to the variance of its center/mode, which is the Mann--Whitney estimator $\hat\theta_{m,n}$.  Under the additional assumption that 
\[ \lambda = \lim_{m,n \to \infty} \frac{m}{m+n} \in (0,1), \]
\citet[][Theorem~7.3]{hoeffding1948} showed that the asymptotic variance of $\hat\theta_{m,n}$ is 
\begin{equation}
\label{eq:mh.variance}
\frac{1}{m+n} \Bigl( \frac{\tau_{10}}{\lambda}+\frac{\tau_{01}}{1-\lambda} \Bigr), 
\end{equation}
where 
\[ \tau_{10} = \cov\{1(U_1 > V_1), 1(U_1 > V_2)\} \quad \text{and} \quad \tau_{01} = \cov\{1(U_1 > V_1), 1(U_2 > V_1)\}, \]
with $\cov$ the covariance operator under joint distribution $\prob$.  If we take the flat prior in our Gibbs posterior construction, then choosing
\begin{equation}
\label{eq:optimal_omega}
\hat\omega_n = \frac{m+n}{2mn}
\Bigl(\frac{\hat\tau_{10}}{\lambda}+\frac{\hat\tau_{01}}{1-\lambda} \Bigr)^{-1}, 
\end{equation}
with the obvious estimates 
\begin{align*}
\hat\tau_{10} & = \frac{2}{mn(n-1)} \sum_{i=1}^m \sum_{j \neq j'} 1(U_i > V_j) \, 1(U_i > V_{j'}) - \hat\theta_{m,n}^2 \\
\hat\tau_{01} & = \frac{2}{n m (m-1)} \sum_{j=1}^n \sum_{i \neq i'} 1(U_i > V_j) \, 1(U_{i'} > V_j) - \hat\theta_{m,n}^2, 
\end{align*}
will make the Gibbs posterior variance approximately match the Mann--Whitney estimator variance, thus, approximate calibration.  But note that our numerical results in Section~\ref{S:examples} below are all based on the calibration strategy in Algorithm~1.

\begin{algorithm}[t]
\SetAlgoLined
\KwData{$U_1,\dots,U_m$ and $V_1,\dots,V_n$}
\KwIn{Prior distribution; estimate $\hat\theta_{m,n}$ from \eqref{eq:Mann-Whitney}; bootstrap sample size $B$; tolerance $\eps > 0$; and step sizes $(\kappa_t)$ satisfying \eqref{eq:steps}.}
\KwOut{An estimate of the learning rate, $\hat{\omega}_n$.} 
\medskip 
Generate bootstrap samples $U_1^{(b)},\dots,U_m^{(b)}$ and $V_1^{(b)},\dots,V_n^{(b)}$, for $b=1,\dots,B$.\\
Initialize $\omega^{(1)}$ and set $t=1$.\\
\Repeat{$|\Delta| < \eps$}{
  $\omega=\omega^{(t)}$\;
  \For{$b$ in $1\dots B$ }{
   Calculate $\text{HPD}_\omega^{(b)}$, the $100(1-\alpha)$\% highest Gibbs posterior density credible interval, with learning rate $\omega$, based on the $b^\text{th}$ bootstrap sample. 
  }
  Estimate the coverage probability $\hat c_\alpha(\omega)=
  B^{-1} |\{b: \text{HPD}_\omega^{(b)}\ni\hat{\theta}_{m,n}\}|$\;
  Set $\Delta=\hat c_\alpha(\omega)-(1-\alpha)$\;
  Update $\omega^{(t+1)}=\omega + \kappa_t \, \Delta$\;
  Set $t=t+1$\;
 }
 Return $\hat{\omega}_n=\omega^{(t)}$. 
 \caption{Gibbs posterior calibration}
\end{algorithm}

\section{Numerical examples}
\label{S:examples}

\subsection{Simulation studies}

Since the AUC is invariant when random variables $U$ and $V$ undergo the same monotone increasing transformation, we fix the distribution of $V$ to be standard normal and consider four examples for the distribution of $U$:
\begin{description}
\item[\it Example 1.] $U \sim \nm(2,1)$ and $\theta^\star=0.9214$;
\item[\it Example 2.] $U\sim \sn(3,1,-4)$---skew normal---and $\theta^\star=0.9665$;
\item[\it Example 3.] $U\sim 0.2\nm(-1,1)+0.8\nm(2,0.5^2)$ and $\theta^\star=0.8185$;
\item[\it Example 4.] $U\sim 2-\expo(1)$ and $\theta^\star=0.7895$.
\end{description}
Figure~\ref{fig:Example_setting} provides a visualization of the two densities in each of the four examples.  Note that these four examples capture binormality, a slight violation of binormality, a bimodal case, and one where $U$ and $V$ have different supports.  

\begin{figure}[t]
  \centering
  \subfigure[Example~1]{\includegraphics[width=0.4\textwidth,height=5cm]{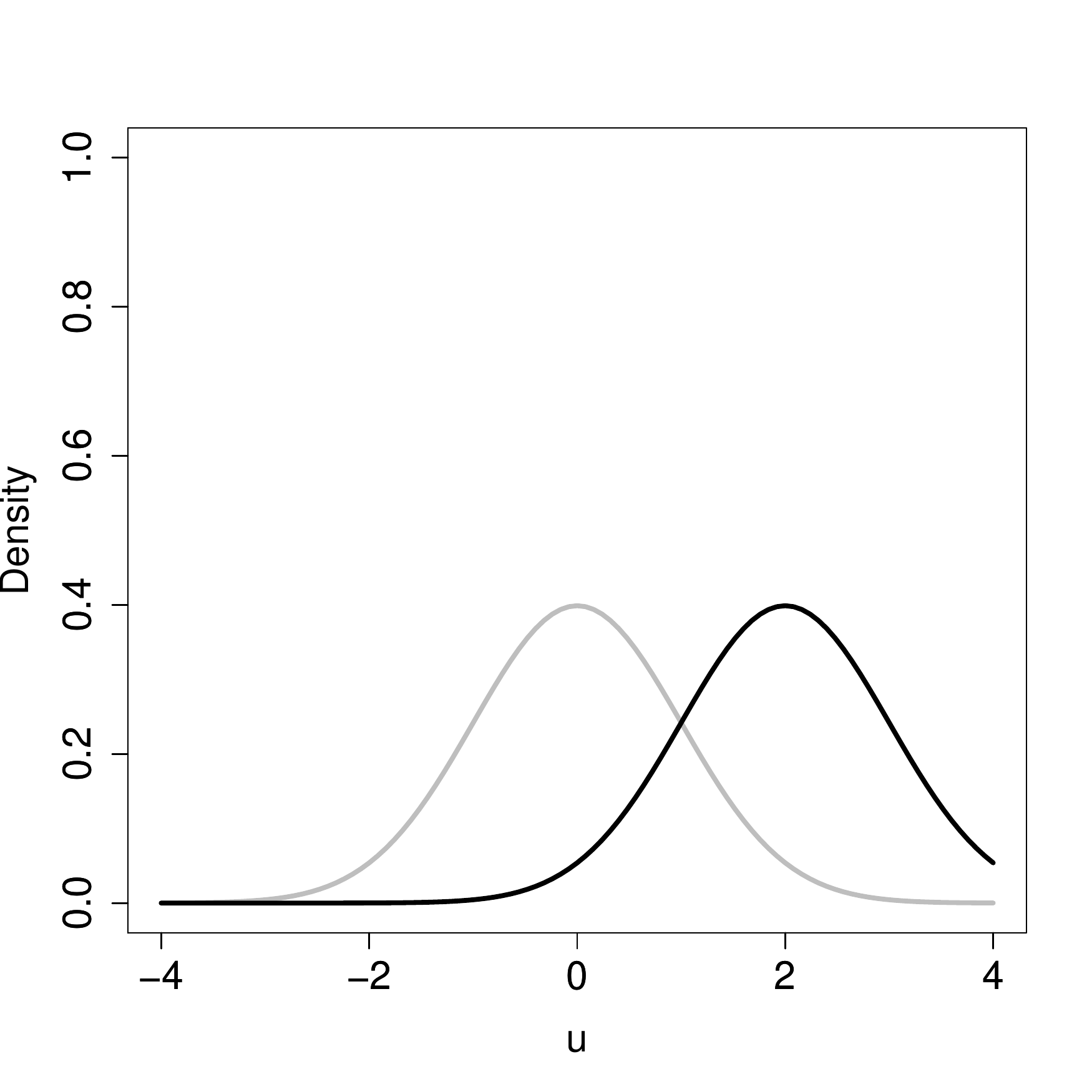}}
  \subfigure[Example~2]{\includegraphics[width=0.4\textwidth,height=5cm]{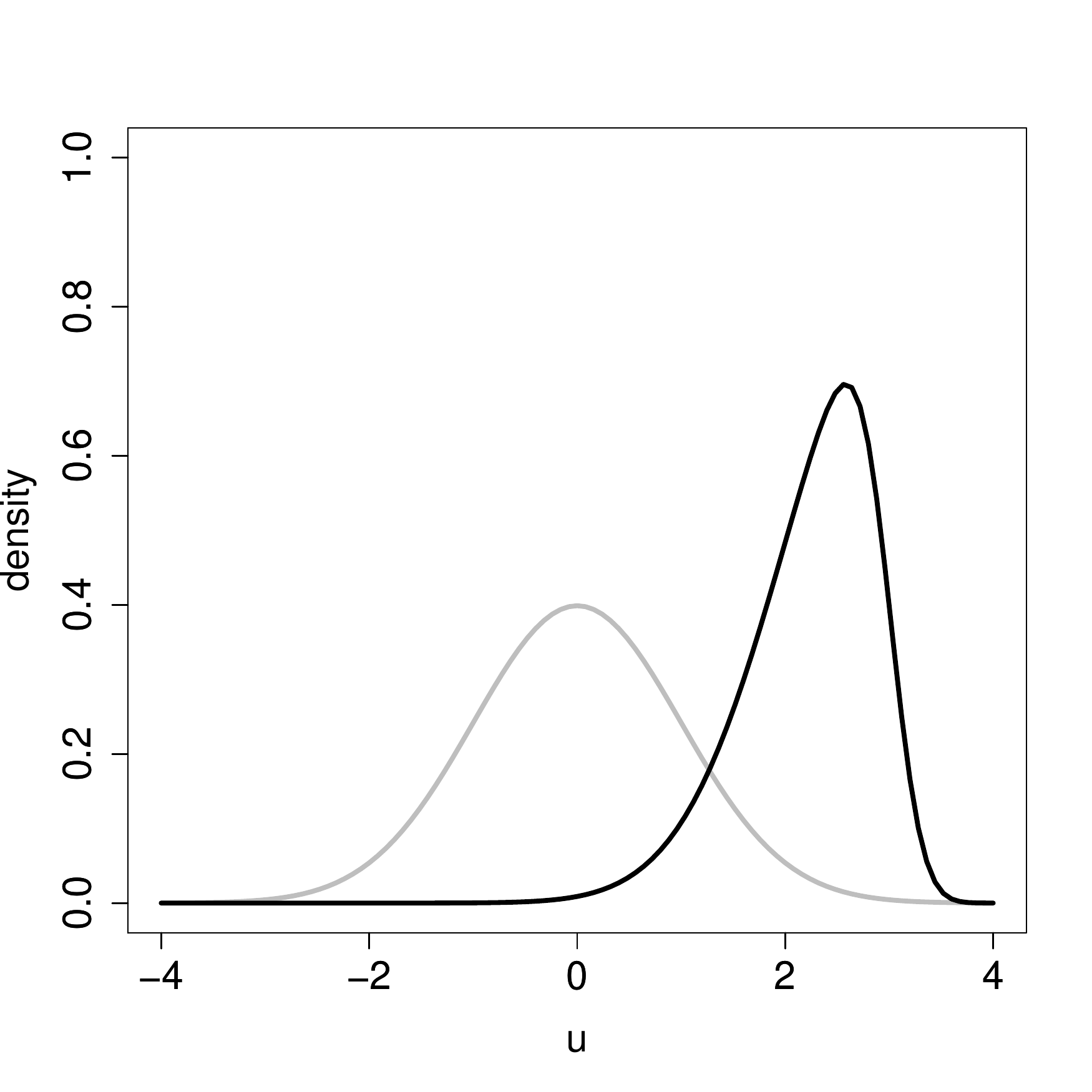}} \\ 
  \subfigure[Example~3]{\includegraphics[width=0.4\textwidth,height=5cm]{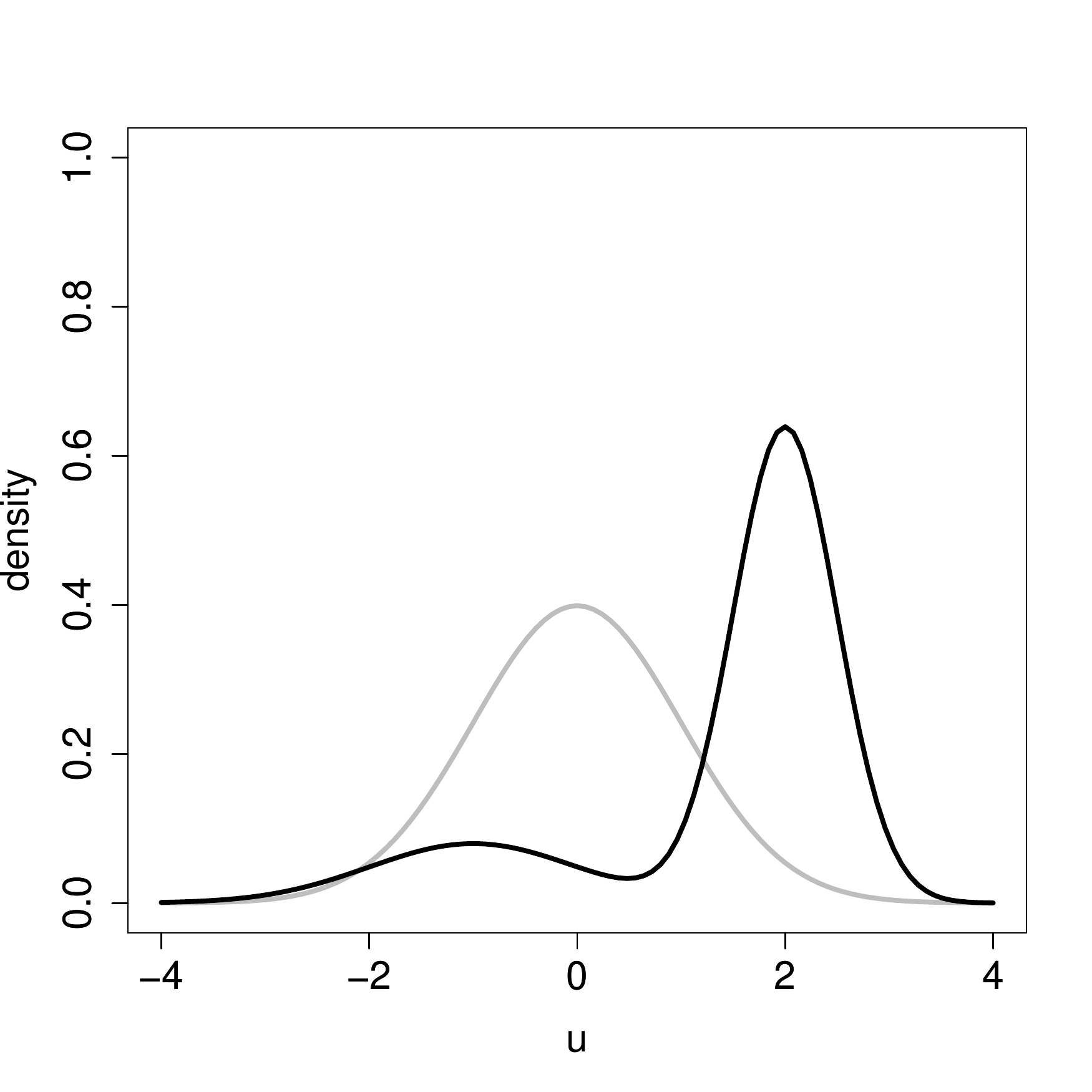}}
  \subfigure[Example~4]{\includegraphics[width=0.4\textwidth,height=5cm]{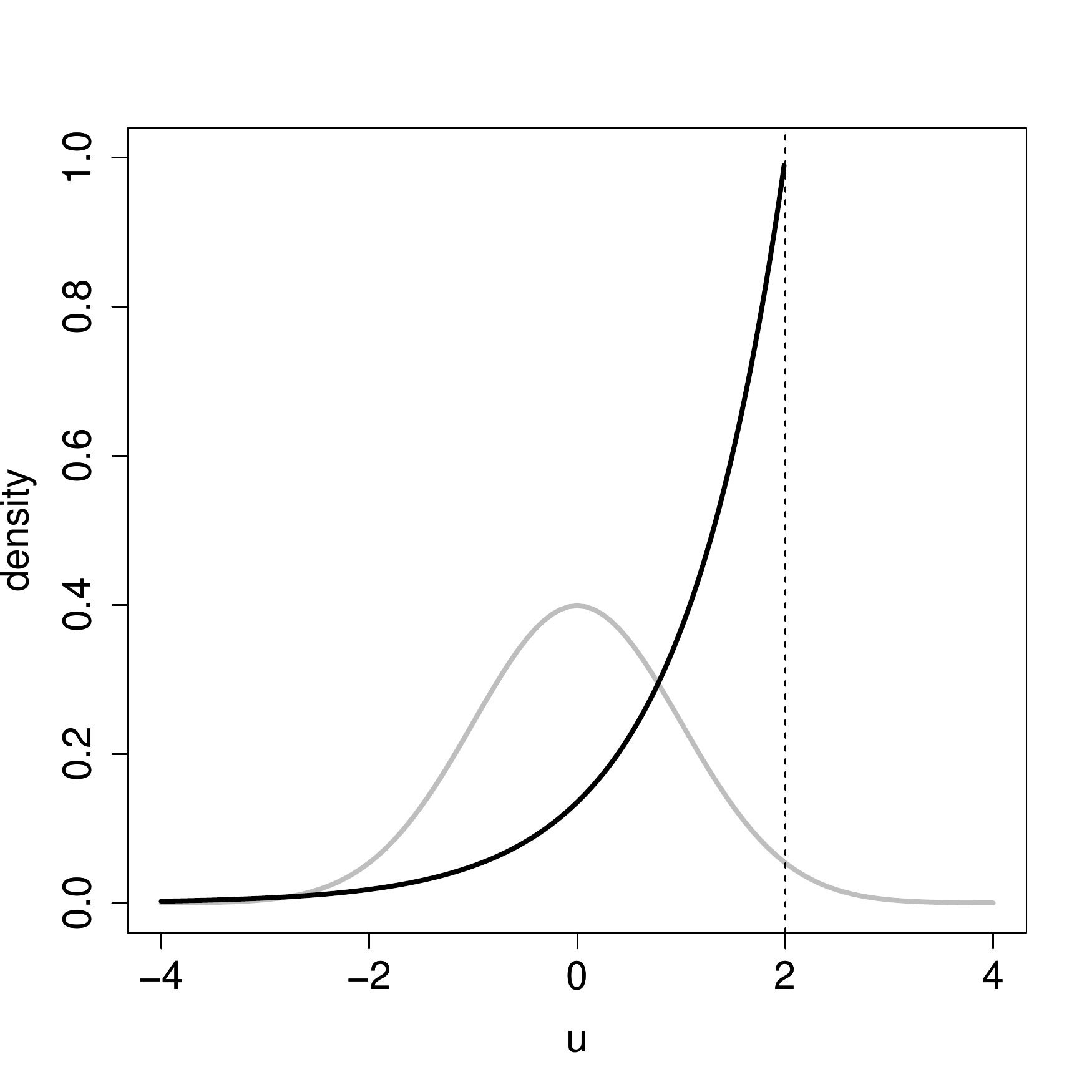}}
  \caption{Density for $V$ (black line) and the standard normal density for $U$ (gray line) in Examples~1--4.}
  \label{fig:Example_setting}
\end{figure}

Here we compare the performance of the Gibbs posterior with the misspecified Bayesian model based on the rank-likelihood (BRL).  For the Gibbs posterior, we use flat prior and follow Algorithm~1, where $B=1000$ bootstrap samples are generated and $\kappa_t=(t+1)^{-0.51}$, which satisfies \eqref{eq:steps}.  For the BRL, 50000 MCMC posterior samples are drawn, with burn-in of 10000.
Tables~\ref{tab:E1}--\ref{tab:E4} present (absolute) bias of the posterior estimator, average posterior standard deviation, average length of $95\%$ credible interval, and the corresponding coverage probability based on $1000$ replications, with increasing observation sizes $m=n=25,50,75,100,125$, for the four examples, respectively.

As can be seen from the bias and standard error columns, both the Gibbs and BRL posteriors concentrate around the true AUC, but the former---thanks to its built-in robustness---tends to have a smaller bias than the latter.  The averaged credible interval length for BRL is slightly smaller than that for the Gibbs posterior, at least when the sample size is large, but at the cost of having unacceptably low coverage probability.  Specifically, for large sample size, the $95\%$ credible intervals from the Gibbs posterior have coverage near the target level 0.95, while the corresponding BRL credible interval tend to under-cover, sometimes severely.  Such a result is also demonstrated in \citet{gu2009bayesian}.  A possible explanation is that the posterior mean of BRL converge to $\theta^\star$ but at a slower speed than the vanishing posterior spread.

\begin{table}[]
\begin{center}
\begin{tabular}{ccccccccc}
\hline
& \multicolumn{2}{c}{Bias} & \multicolumn{2}{c}{Standard Error} & \multicolumn{2}{c}{Mean Length} & \multicolumn{2}{c}{Coverage Prob.} \\ \cline{2-9} 
$n$ & Gibbs & BRL & Gibbs & BRL & Gibbs & BRL & Gibbs & BRL\\ \hline
 25&0.002&0.016&0.035&0.043&0.134&0.165&0.902&0.972\\
 50&0.000&0.007&0.026&0.026&0.103&0.102&0.922&0.931\\
 75&0.000&0.003&0.021&0.020&0.084&0.076&0.939&0.894\\
100&0.000&0.003&0.018&0.016&0.070&0.063&0.935&0.879\\
125&0.001&0.010&0.017&0.014&0.067&0.055&0.940&0.857\\\hline
\end{tabular}
\caption{Gibbs posterior versus BRL for Example~1.}
\label{tab:E1}
\end{center}
\end{table}

\begin{table}[]
\begin{center}
\begin{tabular}{ccccccccc}
\hline
& \multicolumn{2}{c}{Bias} & \multicolumn{2}{c}{Standard Error} & \multicolumn{2}{c}{Mean Length} & \multicolumn{2}{c}{Coverage Prob.} \\ \cline{2-9} 
$n$ & Gibbs & BRL & Gibbs & BRL & Gibbs & BRL & Gibbs & BRL\\ \hline
 25&0.005&0.022&0.020&0.035&0.072&0.132&0.997&0.949\\
 50&0.001&0.006&0.015&0.017&0.058&0.065&0.912&0.904\\
 75&0.000&0.001&0.013&0.012&0.051&0.047&0.919&0.902\\
100&0.000&0.002&0.012&0.010&0.046&0.040&0.931&0.907\\
125&0.000&0.004&0.011&0.009&0.043&0.036&0.944&0.861\\\hline
\end{tabular}
\caption{Gibbs posterior versus BRL for Example~2.}
\label{tab:E2}
\end{center}
\end{table}

\begin{table}[]
\begin{center}
\begin{tabular}{ccccccccc}
\hline
& \multicolumn{2}{c}{Bias} & \multicolumn{2}{c}{Standard Error} & \multicolumn{2}{c}{Mean Length} & \multicolumn{2}{c}{Coverage Prob.} \\ \cline{2-9} 
$n$ & Gibbs & BRL & Gibbs & BRL & Gibbs & BRL & Gibbs & BRL\\ \hline
 25&0.002&0.020&0.065&0.064&0.255&0.246&0.919&0.922\\
 50&0.002&0.016&0.046&0.044&0.180&0.173&0.933&0.900\\
 75&0.000&0.011&0.037&0.035&0.145&0.138&0.921&0.887\\
100&0.000&0.008&0.032&0.030&0.126&0.117&0.936&0.897\\
125&0.001&0.003&0.029&0.027&0.113&0.104&0.934&0.890\\\hline
\end{tabular}
\caption{Gibbs posterior versus BRL for Example~3.}
\label{tab:E3}
\end{center}
\end{table}

\begin{table}[]
\begin{center}
\begin{tabular}{ccccccccc}
\hline
& \multicolumn{2}{c}{Bias} & \multicolumn{2}{c}{Standard Error} & \multicolumn{2}{c}{Mean Length} & \multicolumn{2}{c}{Coverage Prob.} \\ \cline{2-9} 
$n$ & Gibbs & BRL & Gibbs & BRL & Gibbs & BRL & Gibbs & BRL\\ \hline
 25&0.000&0.025&0.066&0.063&0.258&0.243&0.925&0.902\\
 50&0.000&0.024&0.045&0.043&0.176&0.168&0.937&0.844\\
 75&0.001&0.020&0.037&0.033&0.144&0.130&0.930&0.788\\
100&0.000&0.003&0.032&0.028&0.125&0.109&0.942&0.861\\
125&0.000&0.020&0.029&0.026&0.112&0.100&0.938&0.803\\\hline
\end{tabular}
\caption{Gibbs posterior versus BRL for Example~4.}
\label{tab:E4}
\end{center}
\end{table}

Finally, we investigate the learning rate estimates under the Gibbs setting.  Figure~\ref{fig:omega_plot} shows, for each of the four simulation examples, the oracle learning rate (red) compared to those obtained from Algorithm~1.  Recall, from Section~\ref{SS:calibrate}, the oracle learning rate corresponds to exact credibility--coverage matching, so the fact that the estimates based on Algorithm~1 closely follow the oracle is further indication that our Gibbs posterior is properly calibrated to achieve the desired coverage probability. Note, also, that the slope of the red line is roughly $-1$ which, on the log scale, agrees with the tolerable decay rate, $(m \vee n)^{-1}$, suggested by the general theory in Section~\ref{S:converge rate}.  


\begin{figure}[t]
  \centering
  \subfigure[Example~1]{\includegraphics[width=0.42\textwidth,height=5cm]{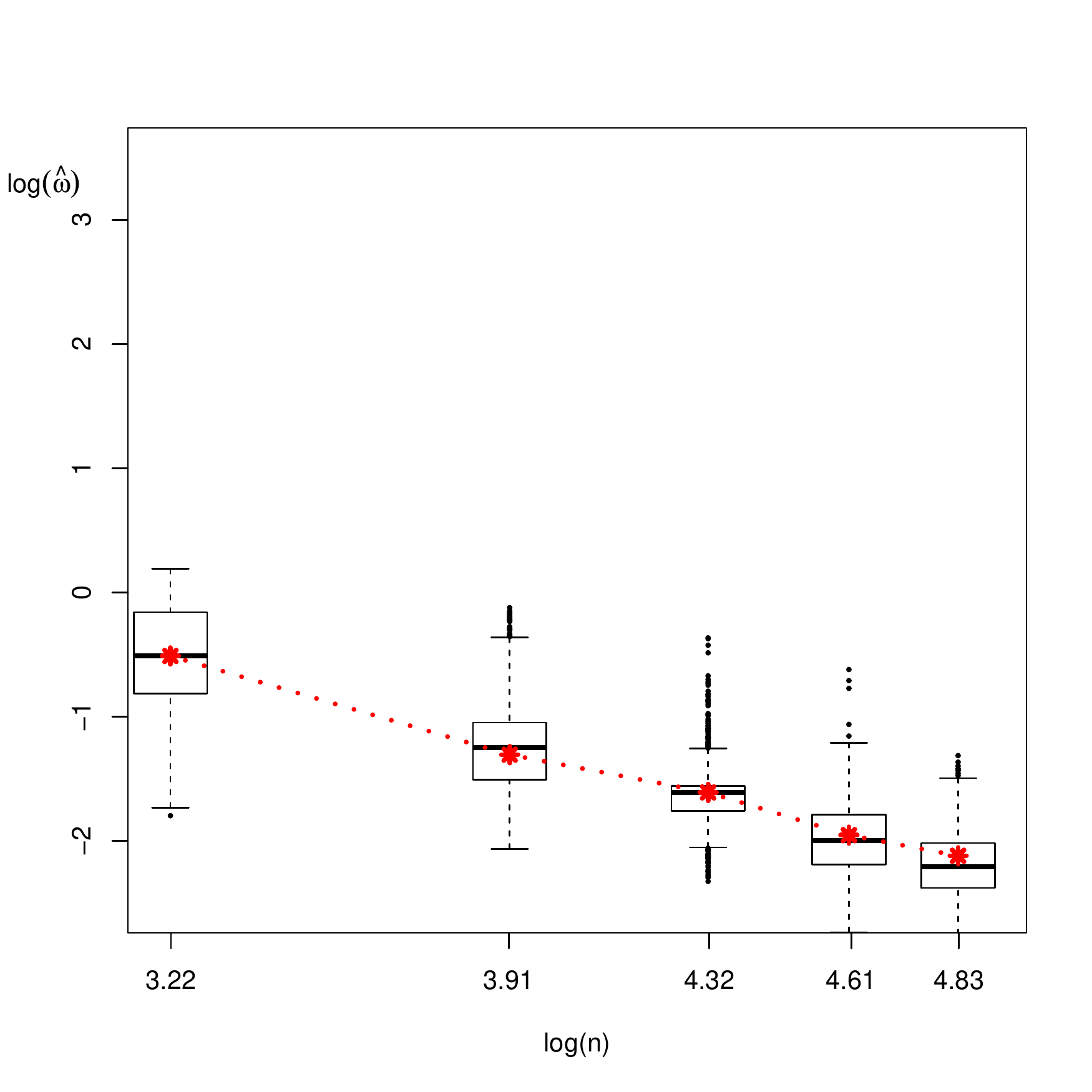}}
  \subfigure[Example~2]{\includegraphics[width=0.42\textwidth,height=5cm]{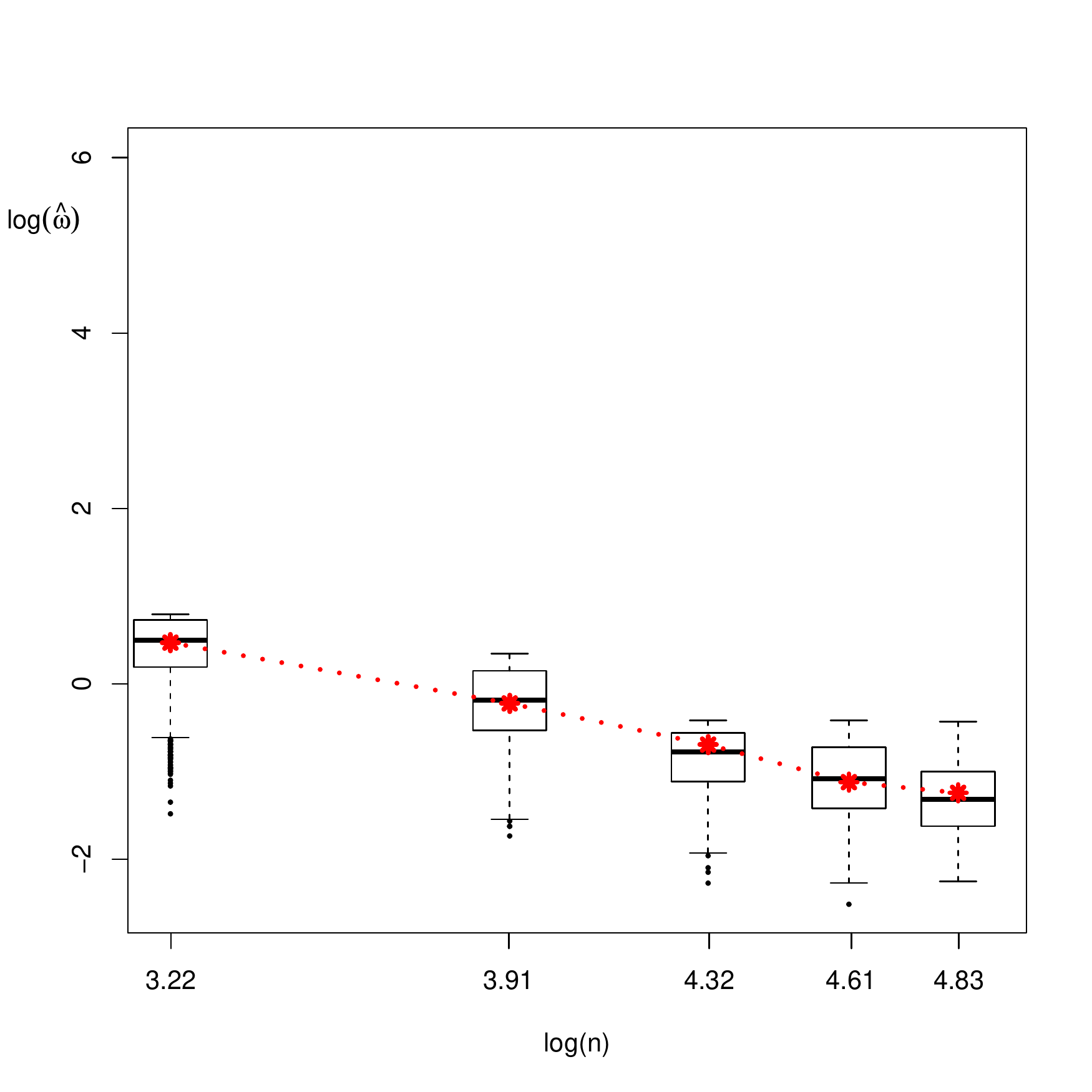}} 
  \subfigure[Example~3]{\includegraphics[width=0.42\textwidth,height=5cm]{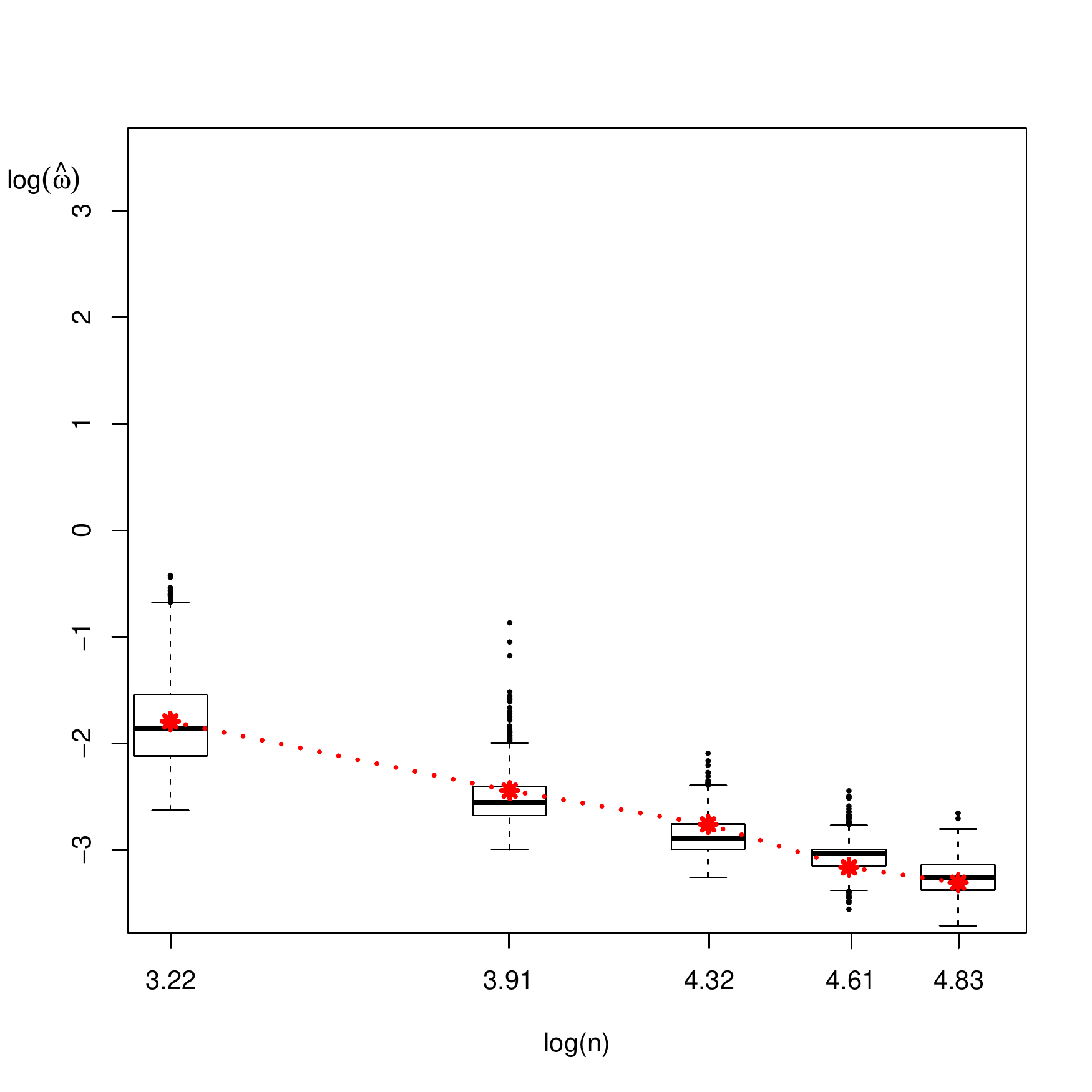}}
  \subfigure[Example~4]{\includegraphics[width=0.42\textwidth,height=5cm]{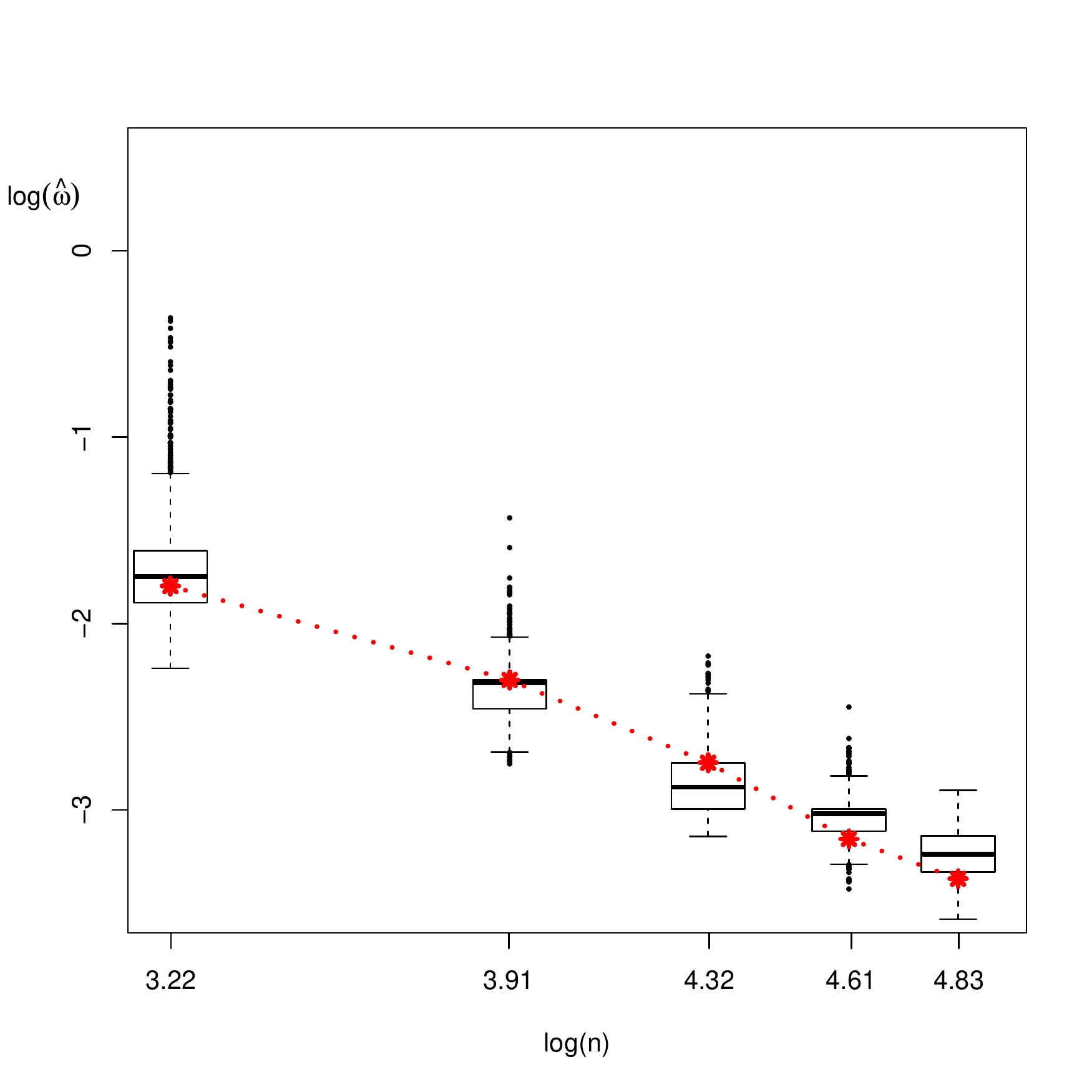}}
 \caption{Boxplots of the learning rate estimates from Algorithm~1 versus sample size, on the log scale, for Examples~1--4, with $n \in \{25, 50, 75, 100, 125\}$.  The red dotted line represents logarithm of the oracle learning rate defined in Section~\ref{SS:calibrate}.}  
  \label{fig:omega_plot}
\end{figure}

\subsection{Real data analysis}

Data consisting of serum measurements of two biomarkers for pancreatic cancer was published by \citet{wieand1989family}; see, also, the R package {\tt logcondens}.  This was a case-control study including $m=90$ subjects from the diseased group and $n=51$ subjects from the non-diseased group. Specifically,
we consider one biomarker, a cancer antigen (CA-125), and evaluate its performance as a classifier to distinguish the case group from the control group.
Table~\ref{fig:realdata} presents results from two Gibbs posteriors and two BRLs.
Gibbs1 and Gibbs2 employ Algorithm 1 with flat prior and truncated normal prior ($\text{location} = 0.75, \text{scale} = 0.9^2$), respectively. 
The two BRLs start the MCMC sampling with different initial values for $(a,b)$, namely, $(2,2)$ for BRL1 and $(3,2)$ for BRL2, respectively, and use $300000$ posterior samples with $5000$ burn-in.
The two Gibbs posteriors have estimates slightly larger than that from two BRLs, with comparable standard errors.  The two BRL credible intervals are slightly shorter than the Gibbs intervals but, in light of the simulation results presented above, especially in the case of relatively large samples like considered here, it is likely that the BRL intervals are ``too short,'' while the Gibbs intervals are not.  

\begin{table}[t]
\begin{center}
\begin{tabular}{ccccc}
\hline
 & Gibbs1 & Gibbs2 & BRL1 & BRL2 \\ \hline 
Posterior mean 
& 0.705 & 0.705 & 0.691 & 0.697 \\
Standard error 
& 0.045 & 0.046 & 0.046 & 0.041 \\
Credible interval
& (0.615, 0.795) & (0.615, 0.796) &  (0.598, 0.774) & (0.612, 0.775)\\ 
Learning rate
& 0.052 & 0.051 &  --- & --- \\\hline
\end{tabular}
\end{center}
\caption{Results for CA-125 based on two Gibbs posteriors and two BRLs.}
\label{fig:realdata}
\end{table}

\ifthenelse{1=1}{}{
\begin{table}[]
\begin{tabular}{cccccc}
\hline
 & Gibbs1 & Gibbs2 & Gibbs3 & BRL1 & BRL2 \\ \hline 
Posterior mean 
& .705 & 0.705 & .706 & .691 & .697 \\
Standard error 
& .045 & 0.046 & .036 & .046 & .041 \\
Credible interval
& (.615, .795) & (.615, .796) & (.636, .776) & (.598, .774) & (.612, .775)\\ 
$\omega$ estimate
& .052 & .051 & .086 & --- & --- \\\hline
\end{tabular}
\caption{Real data analysis for CA-125 for three Gibbs posteriors and two BRLs which takes different prior distributions}
\end{table}
}

\section{Conclusion}
\label{S:discuss}

In certain applications, the parameters of interest can be defined as minimizers of an appropriate risk function, separate from any statistical model.  In such cases, one can avoid potential model misspecification biases by working some kind of ``model-free'' approach.  The present paper considered one such example, namely, inference on the AUC, where the state-of-the-art statistical model is one that depends on an infinite-dimensional nuisance parameter.  As an alternative to switching to rank-based methods that ignore relevant features of the observed data, we propose to construct a Gibbs posterior distribution for direct inference on the AUC, without specifying a model or introducing any nuisance parameters.  This simplifies our computations and prior specifications, while allowing us to avoid potential model misspecification biases without sacrificing on the desirable asymptotic convergence properties.  Moreover, a strategy for tuning the Gibbs posterior's learning rate is recommended, that leads to credible intervals having the nominal frequentist coverage probability.  


A direct extension of our work here is the inference on the analog of AUC in settings that involve three-group classifiers, namely, the volume under the ROC surface, or VUS \citep[e.g.,][]{mossman1999three}.  Similar to the set up here for the AUC, the VUS is defined as $\prob(T>U>V)$, where $T$ is the score for the third group.  Then much of the work presented here can be immediately generalized to the VUS case.  

It would also be worthwhile to explore applications of the Gibbs posterior in other multivariate settings.  One example is inference on multivariate quantiles, which are typically defined as minimizers of some expected loss  \citep[e.g.,][]{chaudhuri1996geometric}, so the construction of a Gibbs posterior is both appealing and relatively simple.

\section*{Acknowledgments}

This work is partially supported by the U.S.~National Science Foundation, DMS--1811802.

\appendix

\section{Proof of Theorem~\ref{thm:rate}}

First, recall that, without loss of generality, we assume $n = m \wedge n$ and $n \to \infty$, which implies that $m=m_n \to \infty$ too.  Next, when $n$ (and, hence, $m$) is large, $\theta \mapsto \exp\{-\omega mn R_{m,n}(\theta)\}$ will blow up around $\theta=\hat\theta_{m,n}$ and, since the prior $\pi$ is fixed---and positive in an interval containing $\theta^\star$ and, hence, $\hat\theta_{m,n}$---the Gibbs posterior will be dominated by the empirical risk term.  Therefore, the prior does not affect the asymptotics so, for simplicity, we present the proof only for the case of a flat prior, $\pi(\theta) \equiv 1$.  

By Chebyshev's inequality and the bias--variance decomposition of mean square error, 
\begin{equation}
\label{eq:chebyshev}
\Pi_{m,n}(\{\theta: |\theta-\theta^\star| > K_n n^{-1/2}\}) \leq \frac{n}{K_n^2} \{ V_{m,n} + (M_{m,n} - \theta^\star)^2 \}, 
\end{equation}
where $M_{m,n}$ and $V_{m,n}$ are the mean and variance of the Gibbs posterior distribution, respectively, and are given by 
\begin{align*}
M_{m,n} & = \hat{\theta}_{m,n}+
\sigma_{m,n}\frac{\phi(A_{m,n})-\phi(B_{m,n})}{\Phi(B_{m,n})-\Phi(A_{m,n})} \\
V_{m,n} & = \sigma_{m,n}^2\Big\{
1+
\frac{A_{m,n}\phi(A_{m,n})-B_{m,n}\phi(B_{m,n})}{\Phi(B_{m,n})-\Phi(A_{m,n})}-
\Big[\frac{\phi(A_{m,n})-\phi(B_{m,n})}{\Phi(B_{m,n})-\Phi(A_{m,n})}\Big]^2
\Big\},
\end{align*}
with $\phi$ and $\Phi$ the $\nm(0,1)$ density and distribution functions, respectively, and 
\begin{align*}
A_{m,n} & = -\mu_{m,n}\sigma_{m,n}^{-1} =-\hat{\theta}_{m,n} (2\omega mn)^{1/2} \\ B_{m,n} & = (1-\mu_{m,n}) \sigma_{m,n}^{-1} =(1-\hat{\theta}_{m,n})(2\omega mn)^{1/2}.
\end{align*}
Since $\hat\theta_{m,n}$ is a consistent estimator of $\theta^\star$ (see below), we clearly have that $A_{m,n} \to -\infty$ and $B_{m,n} \to \infty$, so those ratios involving $\phi$ and $\Phi$ above are all $O_p(1)$.  Then we can immediately conclude that $V_{m,n} = O_p((mn)^{-1})$ which takes care of the variance term.  For the bias term, we first have that $\hat\theta_{m,n}$, the Mann--Whitney statistic, is an unbiased estimator of $\theta^\star$ and its variance is upper-bounded by 
\[ \frac{\theta^\star(1-\theta^\star) (m + n)}{mn}. \]
Therefore, for any $\eps > 0$, there exists a number $L=L_\eps$ such that 
\[ \prob\bigl( n^{1/2}|\hat\theta_{m,n} - \theta^\star| > L \bigr) \leq \eps. \]
To see this, use Chebyshev's inequality and the bound on the variance of $\hat\theta_{m,n}$ to get that the left-hand side above is upper-bounded by 
\[ \frac{\theta^\star(1-\theta^\star) (m + n)n}{L^2mn}. \]
Since $(m+n)/m \leq 2$ and $\theta^\star(1-\theta^\star) \leq 1/4$, we can take $L=L_\eps$ sufficiently large that the previous display is less than $\eps$.  This implies that $|\hat\theta_{m,n} - \theta^\star|$ and, hence, $|M_{m,n} - \theta^\star|$ is $O_p(n^{-1/2})$.  Putting everything together, we have that the right-hand side of \eqref{eq:chebyshev} is 
\[ \frac{n}{K_n^2} \{ O_p((mn)^{-1}) + O_p(n^{-1}) \} = O_p(K_n^{-2}). \]
But since $K_n \to \infty$, we have that the upper-bound in \eqref{eq:chebyshev} converges to 0 in $\prob$-probability as $(m,n) \to \infty$, proving the claim.

\bibliographystyle{apa}
\bibliography{Template}

\end{document}